\documentclass[twocolumn]{article}
\usepackage[utf8]{inputenc}
\usepackage[margin=0.5in]{geometry}
\usepackage[small]{titlesec}
\usepackage{cite}
\usepackage{amsmath,amssymb,amsfonts}
\usepackage{algorithmic}
\usepackage{graphicx}
\graphicspath{{figures/}}
\usepackage{textcomp}
\usepackage{url}
\usepackage{braket}
\usepackage{amsthm}
\newcommand{\Titlestring}{Exploiting Symmetry Reduces the Cost of Training QAOA}

\newcommand{\rev}[1]{{#1}}
\usepackage{booktabs}
\usepackage{multirow}

\newtheorem{theorem}{Theorem}

\newtheorem*{remark}{Remark}

\newcommand{\Gate}[1]{\textsc{#1}}

\newcommand{\zgate}{\Gate{Z}}
\newcommand{\ygate}{\Gate{Y}}
\newcommand{\xgate}{\Gate{X}}

\newcommand{\idgate}{\Gate{I}}

\DeclareMathOperator*{\argmax}{arg\,max}

\usepackage{hyperref}
\usepackage{enumerate}

\title{\Titlestring}
\author{%
    Ruslan Shaydulin$^{*}$ %
    \and Stefan M.\ Wild$^{*}$ %
    }
\date{$^{*}${\small Mathematics and Computer Science Division, Argonne National Laboratory, Lemont, IL 60439, USA}}

\begin{document}

\twocolumn[
  \begin{@twocolumnfalse}
\maketitle
\begin{abstract}
A promising approach to the practical application of  the Quantum Approximate Optimization Algorithm (QAOA) is finding QAOA parameters classically in simulation and sampling the solutions from QAOA with optimized parameters on a quantum computer. Doing so requires repeated evaluations of QAOA energy in simulation. We propose a novel approach for accelerating the evaluation of QAOA energy by leveraging the symmetry of the problem. We show a connection between classical symmetries of the objective function and the symmetries of the terms of the cost Hamiltonian with respect to the QAOA energy. We show how by  considering only the terms that are not connected by symmetry, we can significantly reduce the cost of evaluating the QAOA energy. Our approach is general and applies to any known subgroup of symmetries and is not limited to graph problems. Our results are directly applicable to nonlocal QAOA generalization RQAOA. We outline how available fast graph automorphism solvers can be leveraged for computing the symmetries of the problem in practice. We implement the proposed approach on the MaxCut problem using a state-of-the-art tensor network simulator and a graph automorphism solver on a benchmark of 48 graphs with up to 10,000 nodes. Our approach provides an improvement for $p=1$ on \rev{$71.7\%$} of the graphs considered, with a median speedup of \rev{$4.06$}, on a benchmark where $62.5\%$ of the graphs are known to be hard for automorphism solvers. 
\end{abstract}
\vspace{0.5in}
  \end{@twocolumnfalse}
]

\section{Introduction}

Recent advances in quantum hardware~\cite{Arute2019,ibmqv64} pave the way for demonstrating useful quantum advantage in the coming years. A leading candidate for demonstrating such an advantage is the Quantum Approximate Optimization Algorithm (QAOA)~\cite{farhi2014quantum}, which has been implemented on hardware using up to 23 qubits~\cite{googleqaoaexperimental}. An important challenge associated with practical application of QAOA is the need to find good QAOA parameters, known to be a difficult task in practice~\cite{Shaydulin2019MultistartDOI,Shaydulin2019EvaluatingDOI}. This difficulty is exacerbated by the high cost of obtaining such parameters variationally on a quantum computer because of the long sampling times on quantum hardware implementations and the large number of samples required for accurate energy evaluation~\cite{karalekas2020optimizedforvariational,ushijima2019multilevel}. Although a number of machine learning approaches have been proposed~\cite{verdon2019learning, khairy2019learning}, they only partially address this problem, since training the machine learning models requires large datasets, which are both expensive to obtain from quantum devices and not standardized.

A promising approach to overcoming the challenge of finding good QAOA parameters is finding them purely classically by leveraging the locality of the QAOA algorithm and the recent advances in tensor network methods~\cite{yuriQAOAsim,huang2019alibabacloud} to efficiently (in the number of qubits) evaluate the QAOA energy. In this approach, the value of QAOA energy for a Hamiltonian encoding the target problem with given parameters is evaluated one term at a time; for each term, only the corresponding reverse causal cone is considered~\cite{farhi2014quantum,Streif2020training}. Since the size of the reverse causal cone grows only with QAOA circuit depth, but not with the overall size of the problem, this approach allows one to calculate energies for formally infinitely sized problems. Then a classical optimizer can be used to find QAOA parameters by maximizing the QAOA energy. Although the main motivating application for such advances is finding QAOA parameters, fast evaluation of QAOA energy also has important applications to graph analysis~\cite{Szegedy2019qaoaenergies} and to understanding QAOA potential.

In this paper we introduce a novel method that accelerates the evaluation of QAOA energy by leveraging the symmetries of a problem instance. Classical symmetries of the objective function lead to symmetries of the cost Hamiltonian. 
We show that these symmetries introduce classes of equivalence on the terms of the cost Hamiltonian with regard to QAOA energy, making it sufficient to evaluate only the energy of one term from each equivalence class. For some problems, such symmetries of the objective function may be known a priori or can be computed by using state-of-the-art graph automorphism solvers, which \rev{we demonstrate to be} fast compared with the cost of evaluating the energy. We note that our results apply to any subgroup of symmetries of the problem. Moreover, our results are applicable not only to QAOA but also to its nonlocal generalization RQAOA~\cite{braviyobstacles,braviyobstaclesarxiv,bravyihybrid} because of the need to optimize parameters and evaluate the energy of each term to compute correlations. 

We implement this approach numerically for the MaxCut problem using the state-of-the-art tensor network simulator QTensor~\cite{yuriQAOAsim,qtensor} and using \texttt{nauty}~\cite{McKay201494} for computing the automorphism group of the graph. To illustrate the worst-case performance of the proposed methods, we consider graphs known to be hard for the automorphisms solvers we use, and we show that even in this scenario we achieve a median speedup of \rev{$1.79$} for a single energy evaluation and $p=1$. For graphs that are considered easy for graph automorphism solvers, we observe \rev{a median speedup of $8.96$ and a mean speedup of over $2726$}. In practice, the benefits of using the proposed approach are expected to be larger, since energy is typically evaluated repeatedly (e.g., while optimizing QAOA parameters), and the cost of doing so grows with $p$, and the symmetries are computed only once.

\rev{While the results in this paper are based on the ideas presented in \cite{shaydulinsymm}, our results are distinct in the following three ways. First, we focus on the symmetries of the terms of the cost Hamiltonian and the corresponding redundancies in QAOA energy evaluation, rather than the symmetries of the quantum state and the corresponding symmetries in QAOA outcome probabilities. Second, we demonstrate a novel application of the symmetry analysis to accelerating the training of QAOA purely classically. Third, by leveraging the high-performance tensor network simulator QTensor~\cite{qtensor,yuriQAOAsim}, we are able to consider instances with up to 10,000 nodes / qubits, whereas \cite{shaydulinsymm} considered instances with no more than 22 nodes / qubits.}

The rest of the paper is organized as follows. We review relevant background in Sec.~\ref{sec:background}. Our main result connecting the symmetry of the objective function to the classes of equivalence on the terms of the cost Hamiltonian is presented in Sec.~\ref{sec:symmacceleratesqaoatraining}. The results in %
that section are general and are not limited to graph problems.
In Sec.~\ref{sec:graphaut4symmtopsection} we outline how available fast graph symmetry solvers can be used to accelerate QAOA energy computation. In Sec.~\ref{sec:experiments} we present the numerical results. In Sec.~\ref{sec:discussion} we summarize our conclusions and briefly discuss future work.

\section{Background}\label{sec:background}
We briefly review relevant concepts and point the interested reader to further discussions in the literature.

\subsection{Binary optimization on quantum computers} Consider a function on the Boolean cube, $f:\{0,1\}^n \rightarrow \mathbb{R}$,
and the corresponding optimization problem
\begin{equation}\label{eq:opt_problem}
    \max_{x\in \{0,1\}^n}f(x).
\end{equation}

A Hamiltonian $H\in \mathbb{C}^{2^n\times 2^n}$ is said to faithfully represent the function $f$ if 
\begin{equation}\label{eq:ham_cond}
H\ket{x} = f(x)\ket{x}, \quad \forall x\in \{0,1\}^n.    
\end{equation}

For clarity of presentation, we consider only functions that can be represented as $f(x) = \sum_{i}w_if_i(x)$, $w_i\in\mathbb{R}$, where $f_i:\{0,1\}^n \rightarrow \{0,1\}$ are Boolean functions. For such functions the unique Hamiltonian on $n$ qubits satisfying \eqref{eq:ham_cond} is of the form
\begin{equation}\label{eq:ham_form}
H = H(f) = \sum_{S\subset [n]}\hat{f}(S)\prod_{j\in S}\zgate_j,
\end{equation}
where $\zgate_j$ is a Pauli $\zgate$ operator acting on qubit $j$ and coefficients $\hat{f}(S)$ are given by the Fourier expansion of $f$~\cite{hadfieldrepresentation}. Many objective functions of practical interest are given as Boolean polynomials, with Hamiltonians of the form \eqref{eq:ham_form} following trivially by substituting $x\in \{0,1\} \rightarrow \frac{\idgate - \zgate}{2}$. For a detailed discussion, the reader is referred to \cite{hadfieldrepresentation}.

\subsection{Quantum Approximate Optimization Algorithm} The Quantum Approximate Optimization Algorithm~\cite{farhi2014quantum} is a hybrid quantum-classical algorithm that combines a parameterized quantum evolution with a method for finding good parameters. Applied to the problem \eqref{eq:opt_problem}, QAOA prepares a parameterized quantum state 
\[
\ket{\vec{\beta}, \vec{\gamma}}_p :=
U_B(\beta_p)U_P(\gamma_p)\cdots U_B(\beta_1)U_P(\gamma_1)\ket{s},
\]
where $U_P(\gamma) = e^{-i\gamma H}$ is the phase operator with $H$ encoding the objective function $f$ according to \eqref{eq:ham_cond}, 
$$U_B(\beta)=e^{-i\beta B}=e^{-i\beta\sum \limits_{j=1}^n \xgate_j} = \prod_{j=1}^ne^{-i\beta\xgate_j}$$ 
is the mixing operator with $\xgate_j$ denoting the Pauli $\xgate$ acting on qubit $j$, and $\ket{s}=\ket{+}^{\otimes n}$ is a uniform superposition over computational basis states. 

The parameters $\vec{\beta}$ and  $\vec{\gamma}$ have to be chosen in such a way that, upon measuring the state $\ket{\vec{\beta}, \vec{\gamma}}_p$, the solution of desired quality to the original problem \eqref{eq:opt_problem} is retrieved with high probability. QAOA parameters are typically found variationally by using an outer-loop classical optimizer to approximately maximize $E_p(\vec{\beta}, \vec{\gamma}) = \bra{\vec{\beta}, \vec{\gamma}}_p H\ket{\vec{\beta}, \vec{\gamma}}_p$, with the value of $E_p(\vec{\beta}, \vec{\gamma})$  (commonly referred to as QAOA energy) evaluated either from sampling a quantum computer or purely classically as described below, 
although many alternative approaches have been proposed~\cite{farhi2014quantum,brandao2018fixed-rs,khairy2019learning,verdon2019learning}.

\subsection{Classical training of QAOA} \label{sec:classicalqaoaenergy}

Finding QAOA parameters is typically done by using a classical optimizer to approximately maximize the expectation value of the cost Hamiltonian in QAOA state:
\[
\vec{\beta}^*, \vec{\gamma}^* = \argmax E_p(\vec{\beta}, \vec{\gamma}).
\]

Doing so requires repeated evaluations of $E_p(\vec{\beta}, \vec{\gamma})$. For many classes, this value can be evaluated purely classically with resource requirements growing only with the density of the problem and the number, $p$, of QAOA steps and not the number, $n$, of qubits (i.e., variables in the problem)~\cite{farhi2014quantum}. Here we briefly review the mechanism for doing so; for more details and an in-depth discussion, the reader is referred to \cite{Streif2020training,brandao2018fixed-rs}.

Consider a Hamiltonian of the form \eqref{eq:ham_form}. Then the value $E_p(\vec{\beta}, \vec{\gamma})$ can be written as
\begin{align}
    E_p(\vec{\beta}, \vec{\gamma}) & = \bra{\vec{\beta}, \vec{\gamma}}_p H\ket{\vec{\beta}, \vec{\gamma}}_p \nonumber \\
    & = \bra{\vec{\beta}, \vec{\gamma}}_p \sum_{S\subset [n]}\hat{f}(S)\prod_{j\in S}\zgate_j \ket{\vec{\beta}, \vec{\gamma}}_p  \nonumber \\
    & = \sum_{S\subset [n]}\hat{f}(S) \bra{\vec{\beta}, \vec{\gamma}}_p\prod_{j\in S}\zgate_j\ket{\vec{\beta}, \vec{\gamma}}_p \\
    & = \sum_{S\subset [n]}\hat{f}(S) E_p(\vec{\beta}, \vec{\gamma}, S). \nonumber
\end{align}

By linearity,  each term $E_p(\vec{\beta}, \vec{\gamma}, S)$ (corresponding to observable $\prod_{j\in S}\zgate_j$) can be evaluated independently. To see that each term $E_p(\vec{\beta}, \vec{\gamma}, S)$ can be evaluated with resource requirements independent of $n$, note that 
\begin{align*}
    E_p(\vec{\beta}, \vec{\gamma}, S) =& \bra{\vec{\beta}, \vec{\gamma}}_p\prod_{j\in S}\zgate_j\ket{\vec{\beta}, \vec{\gamma}}_p \\ 
     =& \bra{s}U_P^\dagger(\gamma_1)U_B^\dagger(\beta_1)\cdots \\
    & \cdots U_P^\dagger(\gamma_p)\underbrace{U_B^\dagger(\beta_p) \prod_{j\in S}\zgate_j U_B(\beta_p)}_{O(S)}U_P(\gamma_p)\cdots \\
    & \cdots U_B(\beta_1)U_P(\gamma_1)\ket{s}.
\end{align*}

Now consider $O(S)$. Since  $U_B(\beta_p) = \prod_{j=1}^ne^{-i\beta_p\xgate_j}$, all factors $e^{-i\beta_p\xgate_j}$ with $j\notin S$ will commute through $\prod_{j\in S}\zgate_j$ and cancel out, resulting in operator $O(S)$  having support only on $|S|$ qubits; that is,
\[
O(S) = \prod_{j\in S}e^{i\beta\xgate_j}\prod_{j\in S}\zgate_j\prod_{j\in S}e^{-i\beta\xgate_j}.
\]

From \eqref{eq:ham_form}, the phase separating operator is \begin{align*}
    U_P(\gamma) & = e^{-i\gamma H} \\
    & = e^{-i\gamma \sum \limits_{\tilde{S}\subset [n]}\hat{f}(\tilde{S})\prod \limits_{j\in \tilde{S}}\zgate_j} \\
    & = \prod_{\tilde{S}\subset [n]} e^{-i\gamma \hat{f}(\tilde{S}) \prod_{j\in \tilde{S}}\zgate_j}.
\end{align*}

Analogously to the previous step, whenever $\tilde{S}\cap S = \emptyset$, the factor $e^{-i\gamma \hat{f}(\tilde{S}) \prod_{j\in \tilde{S}}\zgate_j}$ acting on qubits $\tilde{S}$ will commute through and cancel out. Therefore the operator after one step of QAOA  has support only on the set of qubits 
$$\rev{S_{p=1}} = S\cup \{k \in \tilde{S} : \tilde{S}\cap S \neq \emptyset \}.$$ 

Note that the support of the operator after one step of QAOA (\rev{$|S_{p=1}|$ qubits}) does not depend on the overall size of the problem\rev{, $n$, explicitly. Rather, it depends only on} the density of the problem. 
Repeating the above steps, at each step the support of the operator (and thus the complexity of evaluating $E_p(\vec{\beta}, \vec{\gamma}, S)$) grows by including more interacting qubits that are within the \textit{reverse causal cone} of $\prod_{j\in S}\zgate_j$.

For dense instances (e.g.,  where  a term connects any two variables), this approach yields no gains \rev{because the reverse causal cone envelops the entire circuit. Specifically, the applicability of the method is limited by the size of the largest reverse causal cone of any term in the Hamiltonian}. However, for relatively sparse problems with quadratic terms ($|S|=2$, such as MaxCut), this technique has been demonstrated to enable evaluating $E_p(\vec{\beta}, \vec{\gamma})$ and finding good QAOA parameters for instances with hundreds to thousands of nodes with $p\leq 4$~\cite{yuriQAOAsim,huang2019alibabacloud}.

\subsection{MaxCut}

For a graph $G=(V,E)$, the goal of the MaxCut problem is to partition the set of nodes $V$ into two parts such that the number of edges spanning multiple parts is maximized. MaxCut on general graphs is APX-complete~\cite{papadimitriou1991optimization}. The MaxCut objective is encoded by the Hamiltonian
\[
H_{\text{MaxCut}} = \frac{1}{2}\sum_{(u,v)\in E}(I-\zgate_u\zgate_v).
\]

\subsection{Objective function symmetries and QAOA} 

For an objective function $f$ on $n$ bits, a symmetry is a permutation of the $n$-bit strings, $a\in S_{2^n}: x\rightarrow a(x)$, that leaves the objective function unchanged; that is, $f(x) = f(a(x)), \forall x\in \{0,1\}^n$. In this paper we focus on a subgroup of permutations that are variable index permutations (i.e., permutations of $n$ elements), and we denote this subgroup $S_n$. On qubits, the symmetry $a\in S_n$ can be faithfully represented by a unitary 
\begin{equation}\label{eq:AinS_2n}
    A =\sum_{x\in\{0,1\}^n}\ket{a(x)_1\ldots a(x)_{n}}\bra{x_1\ldots x_n},
\end{equation} 
acting as $A\ket{x}=\ket{a(x)}, \forall x\in \{0,1\}^n$. 
\rev{Note that if an objective function has symmetry $a\in S_n$, the Hamiltonian faithfully representing it as specified by \eqref{eq:ham_cond} satisfies $A^\dagger HA = H$, where $A$ is given by \eqref{eq:AinS_2n}. Additionally, for any $A\in S_{2^n}$, the QAOA initial state $\ket{s}$ satisfies $A\ket{s}=\ket{s}$, and, for any $A\in S_n$, the mixing Hamiltonian $B=\sum_{j=1}^nX_j$ satisfies  $A^\dagger BA = B$.}
For a detailed discussion of classical symmetries and their relation to QAOA, the reader is referred to \cite{shaydulinsymm}.

\subsection{Graph symmetries} 

A symmetry of a graph $G=(V,E)$ is a transformation that leaves the graph unchanged. In this paper we  consider only (vertex) automorphisms, defined as permutations of nodes, $\sigma: V\rightarrow V$, that preserve the edges, that is, $(\sigma(v), \sigma(u))\in E$  if and only if  $(u,v)\in E$. \rev{For any label-independent objective function $f$ defined on $G$, an automorphism $\sigma$ of $G$ is a symmetry of $f$.} Random  three-regular and Erdos--Renyi model graphs have been shown to almost surely have no nontrivial automorphisms~\cite{kim2002asymmetry,erdHos1963asymmetric}. At the same time, many graphs of interest for the study of QAOA are highly structured, with the size of a group of automorphisms as large as $n!$~\cite{Szegedy2019qaoaenergies}.

\subsection{Fast solvers for graph symmetries}

In practice, heuristics are used for computing the group of automorphisms of a graph. A number of software implementations exist, including \texttt{nauty}~\cite{mckay1978computing,mckay1981practical} (used in this paper), \texttt{Traces}~\cite{piperno2008search,McKay201494},  \texttt{saucy}~\cite{darga2004exploiting}, \texttt{Bliss}~\cite{junttila2007engineering}, and \texttt{conauto}~\cite{lopez2014novel}. These packages are capable  of solving most graphs with up to tens of thousands of nodes in less than a second~\cite{McKay201494} and thus are sufficient for tackling problems that can realistically be expected to be relevant to noisy intermediate-scale quantum devices. \rev{Any automorphism solver could be used; in Sec.~\ref{sec:experiments} we illustrate the performance obtained on adversarial and more favorable instances by using the solver \texttt{nauty}.}

A 2015 breakthrough by Babai~\cite{babai2015giquasipoly,Babai2016} for the graph isomorphism problem leads to a quasi-polynomial algorithm for computing the group of automorphisms of a graph; no efficient algorithm is known in general. For an overview of complexity-theoretic aspects of obtaining the group of symmetries of a graph, the reader is referred to~\cite[Sec. 3.3]{shaydulinsymm}.

\section{Using symmetry to reduce the number of observables to evaluate}\label{sec:symmacceleratesqaoatraining}

In this section we show how classical symmetries of the objective function reduce the cost of evaluating QAOA energy. The central observation is that variable index permutation symmetries of the objective function induce symmetries of the subsets $S$ of qubit indices in \eqref{eq:ham_form} corresponding to observables (products of Pauli $\zgate$ operators) in the cost Hamiltonian. Since these symmetries impose classes of equivalence on the terms of the Hamiltonian and since
such symmetries are preserved by the QAOA ansatz, the QAOA energy can be computed by evaluating only the value of one observable from each equivalence class.

We now formalize this observation. Consider a permutation of variable (qubit) indices $a\in S_n$. This permutation induces an action on the subsets of qubit indices $\{j_1,\ldots, j_k\}\subset{[n]}$ by the rule $a(\{j_1,\ldots, j_k\}) = \{a(j_1),\ldots, a(j_k)\}$. 
We have the following result. 

\begin{theorem}\label{thm:observablessymm}
Consider a function $f:\{0,1\}^n\rightarrow \mathbb{R}$, its symmetry $a\in S_n$, and depth-$p$ QAOA $\ket{\vec{\beta}, \vec{\gamma}}_p = U_B(\beta_p)U_P(\gamma_p)\cdots U_B(\beta_1)U_P(\gamma_1)\ket{s}$, where $U_P(\gamma)$ is a phase separating operator with the Hamiltonian \rev{$H$} of the form \eqref{eq:ham_form} faithfully representing $f$. Consider two observables  $\prod_{j\in S_1}\zgate_j$ and $\prod_{j\in S_2}\zgate_j$ connected by $a$ in the sense that $a(S_1) = S_2$.

Then, the expectation value of the two observables in the QAOA state is equal; that is,
\[
\bra{\vec{\beta}, \vec{\gamma}}_p\prod_{j\in S_1}\zgate_j\ket{\vec{\beta}, \vec{\gamma}}_p = \bra{\vec{\beta}, \vec{\gamma}}_p\prod_{j\in S_2}\zgate_j\ket{\vec{\beta}, \vec{\gamma}}_p,
\]
for all choices of the number of QAOA steps $p$ and parameters $\vec{\beta}, \vec{\gamma}$. 
\end{theorem}

\begin{proof}
Let $A\in S_{2^n}$ be a unitary faithfully representing $a$ as defined by \eqref{eq:AinS_2n}. If the observables $\prod_{j\in S_1}\zgate_j$ and $\prod_{j\in S_2}\zgate_j$ are connected by $a$, then $A^\dagger\prod_{j\in S_1}\zgate_jA = \prod_{j\in S_2}\zgate_j$. 

\rev{Operator $A$ faithfully represents a symmetry of the objective function $f$ and therefore $[A,H]=0$, which implies $[A,U_P(\gamma)] =0, \forall\gamma$. Moreover, $A$ represents a permutation of qubits $a\in S_n$ and therefore $[A,B]=0$, which implies $[A,U_B(\beta)] =0, \forall\beta$. Combining these two observations gives
\begin{align*}
A\ket{\vec{\beta}, \vec{\gamma}}_p & = AU_B(\beta_p)U_P(\gamma_p)\cdots U_B(\beta_1)U_P(\gamma_1)\ket{s} \\
& = U_B(\beta_p)U_P(\gamma_p)\cdots U_B(\beta_1)U_P(\gamma_1)A\ket{s} \\
& = \ket{\vec{\beta}, \vec{\gamma}}_p,
\end{align*}
where the last equality is from the QAOA initial state being invariant under any permutation of computational basis states: $P\ket{s}=\ket{s}, \forall P\in S_{2^n}$. This} completes the proof.
\end{proof}

The algorithm for computing the QAOA energy is as follows. Suppose we hold a group of classical symmetries of the objective function (in Sec.~\ref{sec:graphaut4symmtopsection} we show how such groups can be heuristically obtained for many relevant problem classes). Begin by computing classes of equivalence $\{\tilde{S}_m\}_{m=1}^{M}$ on subsets of qubit (variable) indices corresponding to the terms of the Hamiltonian. If we denote an (arbitrarily chosen) representative of the equivalence class by $s_m \in \tilde{S}_m$, the total energy can be computed by evaluating $M$ terms as follows:
\[
E_p(\vec{\beta}, \vec{\gamma}) = \sum_{m=1}^{M}|\tilde{S}_m| \bra{\vec{\beta}, \vec{\gamma}}_p\prod_{j\in s_m}\zgate_j\ket{\vec{\beta}, \vec{\gamma}}_p.
\]

Note that although in this paper we focus on the case where the value $\bra{\vec{\beta}, \vec{\gamma}}_p\prod_{j\in s_m}\zgate_j\ket{\vec{\beta}, \vec{\gamma}}_p$ is computed classically, the proposed method is applicable irrespective of how the value of the observable is evaluated.

\begin{remark}
Although Theorem~\ref{thm:observablessymm} is formulated for QAOA and for observables that are products of the Pauli $\zgate$, it can be easily generalized by noting that the relevant conditions are (i) invariance of the ansatz under symmetry (i.e., $A\ket{\psi(\theta)} = \ket{\psi(\theta)}$) and (ii) observables satisfying the relation $A^{\dagger}\prod_{j\in S_1}P_jA = \prod_{j\in S_2}P_j$, where $P_j\in \{\xgate_j,\ygate_j,\zgate_j\}$ is a single-qubit Pauli acting on qubit $j$.  
\end{remark}
An interesting future direction is exploring the applications of such generalizations to reducing the resource requirements of the variational quantum eigensolver~\cite{peruzzo2014variational-rs}.

\begin{figure*}[tbh]
    \centering
    \includegraphics[width=0.85\textwidth]{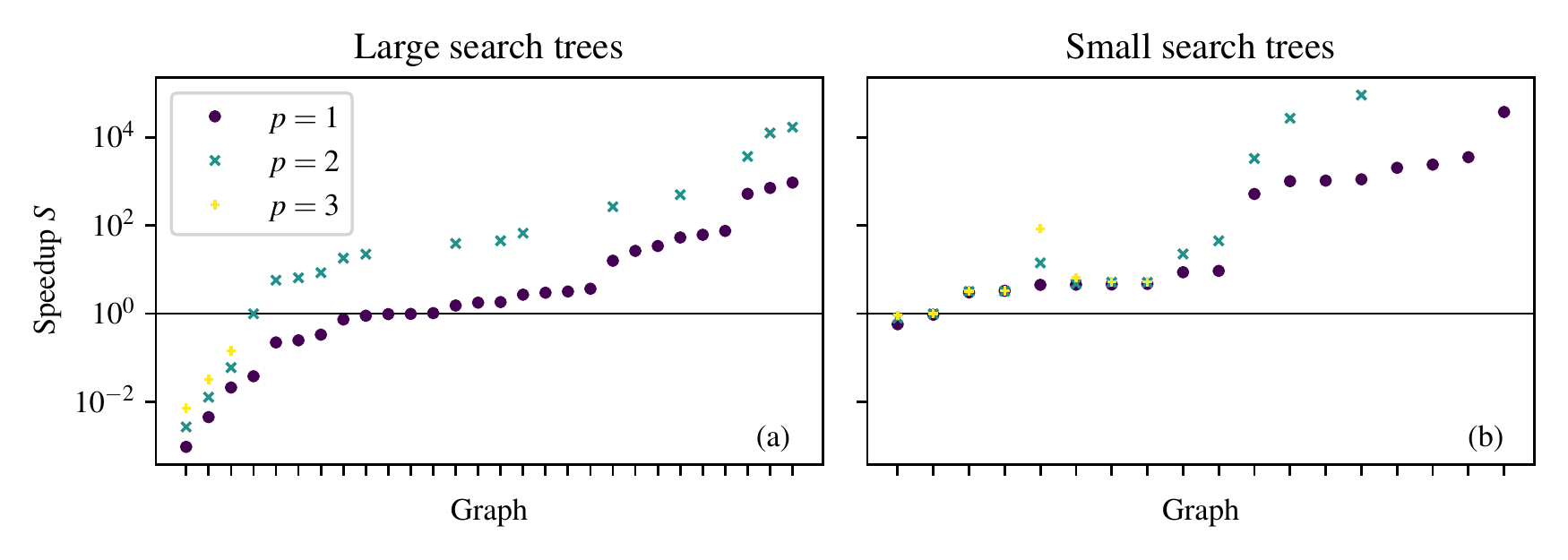}
    \caption{Speedup obtained by leveraging the symmetries of the graph for the graphs with large search trees (a) and small search trees (b). The horizontal axis is the graphs in the benchmark set of \cite{piperno2008search}, sorted in order of increasing speedup. For some graphs (indicated with a missing marker for $p=2$ or $p=3$), the energy could not be computed given the chosen resource constraints. For $p=1$, the improvement is observed for $89\%$ of the graphs with a small search tree and \rev{$61\%$} of the graphs with a large search tree. For the detailed results, see Tables~\ref{tab:dflarge} and  \ref{tab:dfsmall}.}
    \label{fig:speedup}
\end{figure*}

\section{Accelerating QAOA training by using fast graph automorphism solvers} \label{sec:graphaut4symmtopsection}
To apply the findings of Sec.~\ref{sec:symmacceleratesqaoatraining} in practice, one needs a way of obtaining the group of symmetries of the classical objective function. We now outline how heuristic graph automorphism solvers can be used for this purpose, both for problems defined on unweighted graphs and in more general cases.

The first class of problems we consider are problems defined on unweighted graphs. Graph problems are of particular interest to the study and near-term applications of QAOA and have been studied extensively in recent years. Problems considered  include MaxCut~\cite{braviyobstacles,crooks2018performance,zhou2018qaoaperformance}, Maximum-$k$-Cut~\cite{bravyihybrid}, Maximum Independent Set~\cite{farhi2020qaoaneedsfullgraphtypical,farhi2020qaoaneedsfullgraphworst}, Community Detection~\cite{shaydulin2018community-rs,Shaydulin2019MultistartDOI,Shaydulin2019network}, Graph Vertex $k$-Coloring~\cite{do2020planning}, Maximum $k$-Colorable Subgraph~\cite{Wang2020},  Graph Partitioning~\cite{ushijima2019multilevel}, and many more~\cite{hadfield2017quantum,hadfield2018quantum,shaydulin2019hybrid}. The combination of hardness and sparsity make graph problems especially appealing as an early application of QAOA, as evidenced by the fact that a number of recent experimental demonstrations apply QAOA to graph problems~\cite{googleqaoaexperimental,Lacroix2020}. For problems defined explicitly on unweighted graphs, the group of automorphisms of the graph is a subgroup of the group of symmetries of the problem.

A second class of problems are problems with objective functions of the form
\begin{equation}\label{eq:ising_ham}
f(z) = \sum_{(i,j)\in E}J_{i,j}z_iz_j, \quad z_i\in \{-1,1\}.    
\end{equation}
An important problem class of this form is the  Sherrington--Kirkpatrick spin model. In this case, classical symmetries of the objective function can be found by considering a graph $G=(V,E)$ with each edge in $E$ assigned a color in such a way that two edges $(i,j)$ and $(u,v)$ have the same color if and only if $J_{i,j} = J_{u,v}$. Although standard tools such as \texttt{nauty} and \texttt{Traces} cannot handle graphs with colored edges directly, the problem can be converted to an equivalent one where the colors are assigned to vertices instead of edges~\cite{nautytracesmanual}. In many cases, the coefficients $J_{i,j}$ have discrete support, such as $J_{i,j}\in \{-1,1\}$, and thus the problem can have nontrivial amounts of symmetry. On the other hand, if coefficients come from a continuous distribution, such as a normal distribution truncated to some interval, the problem is expected to have no classical symmetries.

In this work we do not consider problems with higher-order terms. The symmetries of such problems can be found using graph automorphism solvers by considering an appropriately constructed bipartite graph. Because of the overhead of doing so and the lack of readily apparent applications,  we do not examine problems with terms beyond quadratic.

\subsection{Application to RQAOA}

In this subsection, we briefly describe how the proposed approach can be applied to accelerating the recursive quantum approximate optimization algorithm (RQAOA). RQAOA has been proposed in \cite{braviyobstacles,braviyobstaclesarxiv} and aims to overcome the limitations of QAOA that arise from its locality. We begin by outlining the algorithm and then show how our findings can be used to improve its performance in practice. For detailed discussion of RQAOA the reader is referred to \rev{\cite{braviyobstacles,braviyobstaclesarxiv,bravyihybrid}; pseudocode for RQAOA is given in \cite[Appendix C.3]{braviyobstaclesarxiv}.}

Consider an objective of the form \ref{eq:ising_ham} and the corresponding Hamiltonian $H_n = \sum_{(i,j)\in E}J_{i,j}\zgate_i\zgate_j$, where $\zgate_j$ is the Pauli $\zgate$ on qubit $j$. RQAOA begins by finding parameters that maximize the standard QAOA energy $\bra{\vec{\beta}, \vec{\gamma}}_p H_n\ket{\vec{\beta}, \vec{\gamma}}_p$. Then for every term in the Hamiltonian the value $E_p(\vec{\beta}, \vec{\gamma}, (i,j)) = \bra{\vec{\beta}^*, \vec{\gamma}^*}_p \zgate_i\zgate_j\ket{\vec{\beta}^*, \vec{\gamma}^*}_p$ is computed, where $\vec{\beta}^*, \vec{\gamma}^*$ are the optimized parameters. 

Then, the term such that $E_p(\vec{\beta}, \vec{\gamma}, (i,j))$ has the largest magnitude is chosen and variable $z_j$ is eliminated by imposing constraint $z_j = \textrm{sgn}(E_p(\vec{\beta}, \vec{\gamma}, (i,j))z_i$. The process is repeated until the number of variables becomes small enough for the resulting optimization problem to be solved by a classical algorithm or brute force.

Our findings can be applied to RQAOA in two ways. First, at each step the repeated evaluations of QAOA energy needed for finding QAOA parameters can be accelerated by leveraging symmetries as outline above. Second, out of all terms $E_p(\vec{\beta}, \vec{\gamma}, (i,j)), (i,j) \in E$ only those that are not connected by symmetry need to be independently evaluated.

\section{Numerical Experiments}\label{sec:experiments}

We implement the proposed approach on the MaxCut problem using the state-of-the-art tensor network simulator QTensor~\cite{yuriQAOAsim,qtensor} for evaluating QAOA energies and the \texttt{nauty}~\cite{McKay201494} solver for computing the group of automorphisms of a graph. We do not leverage any properties of the MaxCut problem except the fact that it is defined on unweighted graphs, and we expect similar results if our approach is applied to other graph problems. We release the complete dataset and the source code at \cite{codeanddata}.

To illustrate the potential as well as the limitations of the proposed approach, we consider the set of graphs used for benchmarking graph automorphism solvers in \cite{piperno2008search}. The set contains graphs with up to 10,000 nodes and up to 120,050 edges. Following \cite{piperno2008search}, we consider two groups of graphs: graphs with large search trees (Table~\ref{tab:dflarge}, typically hard for graph automorphism solvers) and graphs with small search trees (Table~\ref{tab:dfsmall}, typically easy). We execute \texttt{nauty} in sparse mode via the \texttt{dreadnaut} interface and allow it 72 hours to complete. We note that the observed running times of \texttt{nauty} are slower than those reported in \cite{piperno2008search} and can be further improved by careful tuning, which we purposefully avoided.

To evaluate QAOA energies, we run QTensor  either on each term of the problem Hamiltonian ($t_\textrm{s}$) or only on one term from each class of equivalence computed from the symmetries of the problem ($t_\textrm{acc}$). We limit the total running time to 30 hours and 62 GB of RAM, and we exclude the graphs for which the energy cannot be computed within these constraints even for $p=1$ (typically, the computation of QAOA energy is memory bound). Thus we are left with 30 graphs with large search trees and 18 graphs with small search trees.

\begin{figure}[tb]
    \centering
    \includegraphics[width=0.9\linewidth]{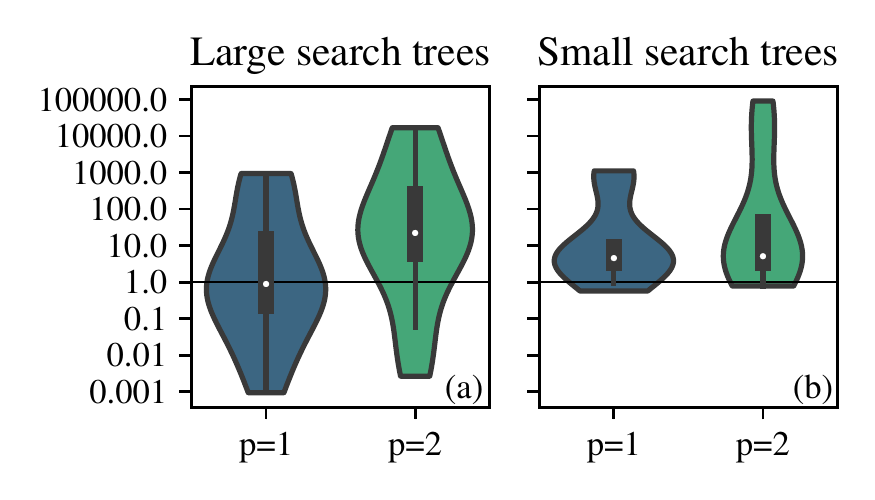}
    \caption{Violin plot of speedup obtained by leveraging the symmetries of the graph for the graphs with large search trees (a) and small search trees (b) for $p=1,2$. We exclude graphs for which energy calculation for $p=2$ could not be completed. \rev{The lack of speedup on some instances is due to the cost of computing the group of automorphisms outweighing the savings in QAOA energy computation from leveraging the computed symmetries.} We observe that the speedup is higher for $p=2$ than it is for $p=1$ for all graphs in our dataset with the exception of \rev{one graph (\texttt{cfi-20}) with a small search tree}. For \rev{this graph} the times to compute the group of automorphism and the energy were low, resulting in small variations in running times and leading to a slight decrease in the speedup (see corresponding row of Table~\ref{tab:dfsmall} for  detailed information).}
    \label{fig:p1vsp2}
\end{figure}

In Fig.~\ref{fig:speedup} we present the speedup achieved by leveraging the symmetries. For the graphs with large search trees (i.e., graphs that are typically hard for automorphism solvers) for $p=1$ we observe an improvement in \rev{17} out of 28 graphs (\rev{$61\%$}), with an overall median speedup of \rev{$1.79$}. For the graphs with small search trees (typically easy for automorphism solvers) for $p=1$ we observe an improvement in 16 out of 18 graphs ($89\%$), with a median speedup of \rev{$8.96$}. For $p>1$ the observed speedup is larger because the fixed cost of computing the group of automorphisms is amortized by the increasing cost of computing the QAOA energy. In this sense, the results for $p=1$ and graphs with large search trees provide a ``worst-case scenario'' for the proposed methods; choosing different graphs and higher depths is expected to lead to higher speedups. This effect is illustrated in Fig.~\ref{fig:p1vsp2}. We do not list statistics such as median speedup for $p>1$ since for many graphs the energy could not be computed given the limit of 62 GB of RAM. The complete results are presented in Tables~\ref{tab:dflarge} and \ref{tab:dfsmall}. 

\rev{For $p=2$ (and graphs \texttt{ag2-49}, \texttt{pg2-32}, \texttt{pp-27-1}) and for $p=1$ (and graph \texttt{had-232}), the standard approach did not complete within 20 hours. For these instance the running time $t_s$ for the standard approach is estimated from actual running time $t_{\text{actual}}\approx 20$h and the number of processed Hamiltonian terms $N_t$ by assuming that the average cost of processed, randomly selected terms is equal to the average cost of processing all Hamiltonian terms. This gives an estimate of the running time $t_s = \frac{t_{\text{actual}} |E|}{N_t}$.}

Note that the running times and the speedup we report are for computing the QAOA energy \emph{once}. In practice, the symmetries would be computed once and be used to evaluate the QAOA energy repeatedly (e.g., in the process of optimizing QAOA parameters). Accounting for such repeated evaluation would dramatically improve the speedup: in the case of the graph with the smallest speedup (largest slowdown), \texttt{mz-aug2-22}, amortizing the cost of computing the automorphisms over 5,000 energy evaluations results in a speedup of $1.37$, rather than the reported slowdown of \rev{$\approx 1061$} observed if the energy is evaluated only once.

\begin{table*}
\centering{\scriptsize
\begin{tabular}{lrrrrrrrrrrrrr}
\toprule
\multirow{2}{*}{Name} &     \multirow{2}{*}{$|E|$} &    \multirow{2}{*}{$|V|$} & \multirow{2}{*}{$N_\textrm{orb}$} & \multirow{2}{*}{$t_\textrm{aut}$} &  \multicolumn{3}{c}{$p=1$}  &  \multicolumn{3}{c}{$p=2$}  & \multicolumn{3}{c}{$p=3$}  \\
\cmidrule(lr){6-8} \cmidrule(lr){9-11} \cmidrule(lr){12-14}
     &           &          &        &           & $t_\textrm{s}$        & $t_\textrm{acc}$  &  $S$       &         $t_\textrm{s}$ &            $t_\textrm{acc}$ & $S$ &  $t_\textrm{s}$  & $t_\textrm{acc}$        &  $S$    \\
\midrule
mz-aug2-22   &     836 &   528 &       419 &           9759.90 &            9.20 &                 4.76 &      0.0009 &           26.09 &                13.48 &      0.0027 &           69.47 &                33.39 &      0.0071 \\
mz-aug2-20   &     760 &   480 &       381 &           1884.47 &            8.41 &                 4.32 &      0.0045 &           23.96 &                12.23 &        0.01 &           61.00 &                30.08 &        0.03 \\
mz-aug2-18   &     684 &   432 &       343 &            354.15 &            7.49 &                 3.93 &        0.02 &           21.67 &                11.07 &        0.06 &           54.61 &                27.08 &        0.14 \\
pp-16-19     &    4641 &   546 &        15 &           6300.20 &          237.46 &                 0.74 &        0.04 &         6226.05 &                20.70 &        0.98 &             nan &                  nan &         nan \\
pp-16-15     &    4641 &   546 &        25 &           1072.21 &          236.93 &                 1.24 &        0.22 &         6302.65 &                34.46 &        5.70 &             nan &                  nan &         nan \\
pp-16-8      &    4641 &   546 &         6 &            944.93 &          233.34 &                 0.30 &        0.25 &         6165.44 &                 8.24 &        6.47 &             nan &                  nan &         nan \\
pp-16-17     &    4641 &   546 &        20 &            719.64 &          238.60 &                 1.00 &        0.33 &         6300.98 &                27.38 &        8.43 &             nan &                  nan &         nan \\
pp-16-11     &    4641 &   546 &         9 &            329.76 &          241.75 &                 0.45 &        0.73 &         6172.57 &                12.33 &       18.04 &             nan &                  nan &         nan \\
pp-16-21     &    4641 &   546 &         8 &            271.77 &          242.15 &                 0.41 &        0.89 &         6273.83 &                11.00 &       22.19 &             nan &                  nan &         nan \\
sts-sw-79-11 &   58539 &  1027 &     58539 &             20.13 &        23826.97 &             24552.61 &        0.97 &             nan &                  nan &         nan &             nan &                  nan &         nan \\
sts-sw-21-10 &     945 &    70 &       945 &              0.03 &           68.45 &                69.78 &        0.98 &             nan &                  nan &         nan &             nan &                  nan &         nan \\
sts-sw-55-1  &   19305 &   495 &     19305 &              2.99 &         4894.36 &              4791.97 &        1.02 &             nan &                  nan &         nan &             nan &                  nan &         nan \\
pp-16-4      &    4641 &   546 &         6 &            154.23 &          234.63 &                 0.30 &        1.52 &         6315.73 &                 8.92 &       38.71 &             nan &                  nan &         nan \\
had-sw-112   &   25312 &   448 &     12768 &            817.91 &        12607.74 &              6324.41 &        1.77 &             nan &                  nan &         nan &             nan &                  nan &         nan \\
pp-16-9      &    4641 &   546 &         6 &            130.35 &          237.46 &                 0.29 &        1.82 &         6196.43 &                 8.31 &       44.69 &             nan &                  nan &         nan \\
pp-16-2      &    4641 &   546 &         4 &             88.16 &          236.78 &                 0.20 &        2.68 &         6227.34 &                 5.58 &       66.43 &             nan &                  nan &         nan \\
sts-67       &   35376 &   737 &     11792 &             22.89 &        11634.51 &              3902.93 &        2.96 &             nan &                  nan &         nan &             nan &                  nan &         nan \\
had-sw-88    &   15664 &   352 &      4004 &            354.12 &         5718.32 &              1455.20 &        3.16 &             nan &                  nan &         nan &             nan &                  nan &         nan \\
had-sw-32-1  &    2112 &   128 &       554 &              1.90 &          228.59 &                60.74 &        3.65 &             nan &                  nan &         nan &             nan &                  nan &         nan \\
pp-16-6      &    4641 &   546 &         6 &             15.33 &          246.42 &                 0.30 &       15.77 &         6241.36 &                 8.20 &      265.24 &             nan &                  nan &         nan \\
had-232      &  108112 &   928 &        62 &           6081.99 &        1.63e+05 &                96.04 &       26.37 &             nan &                  nan &         nan &             nan &                  nan &         nan \\
had-184      &   68080 &   736 &        50 &           1990.91 &        69607.04 &                53.18 &       34.05 &             nan &                  nan &         nan &             nan &                  nan &         nan \\
pp-16-7      &    4641 &   546 &         6 &              4.17 &          237.38 &                 0.30 &       53.15 &         6205.46 &                 8.32 &      496.89 &             nan &                  nan &         nan \\
had-100      &   20200 &   400 &        54 &            117.83 &         8619.30 &                23.14 &       61.14 &             nan &                  nan &         nan &             nan &                  nan &         nan \\
had-52       &    5512 &   208 &        30 &              7.71 &         1002.98 &                 5.69 &       74.81 &             nan &                  nan &         nan &             nan &                  nan &         nan \\
pp-16-1      &    4641 &   546 &         1 &              0.40 &          232.06 &                 0.05 &      520.04 &         6261.13 &                 1.31 &     3660.18 &             nan &                  nan &         nan \\
pp-25-1      &   16926 &  1302 &         1 &              1.80 &         1331.24 &                 0.08 &      707.23 &        86841.64 &                 5.14 &    12503.95 &             nan &                  nan &         nan \\
pp-27-1      &   21196 &  1514 &         1 &              1.86 &         1818.38 &                 0.08 &      935.98 &        1.48e+05 &                 6.91 &    16872.39 &             nan &                  nan &         nan \\
mz-aug2-30   &    1140 &   720 &       nan &               nan &           12.53 &                  nan &         nan &           35.68 &                  nan &         nan &           91.35 &                  nan &         nan \\
mz-aug2-50   &    1900 &  1200 &       nan &               nan &           20.81 &                  nan &         nan &           58.89 &                  nan &         nan &          151.95 &                  nan &         nan \\
\bottomrule
\end{tabular}}
\caption{Benchmark graphs with large search trees and thus hard for graph automorphism solvers. $N_\textrm{orb}$ is the number of classes of equivalence on terms of the cost Hamiltonian, which is equal to the number of edge orbits of the graph. \rev{The smaller the value of $N_\textrm{orb}$, the more symmetric the graph is. For $N_\textrm{orb}=1$, the expected speedup is $O(|E|)$, since we only need to evaluate one term instead of $|E|$.} $t_\textrm{aut}$ is the time (in seconds) to compute the edge orbits of the graph using \texttt{nauty}. $t_\textrm{s}$ is the time (in seconds) to compute the QAOA energy using the standard approach. $t_\textrm{acc}$ is the time (in seconds) to compute the QAOA energy by leveraging the information about edge orbits. $S=\frac{t_\textrm{s}}{t_\textrm{aut}+t_\textrm{acc}}$ is the speedup (larger is better; $S<1$ corresponds to slowdown). ``nan'' denotes cases where the energy computation could not be completed under chosen resource limits. Note that the speedup increases with $p$ because computing the group of symmetries of a graph is a fixed cost, which is offset by the increased cost of evaluating the QAOA energy. For two graphs, mz-aug2-30 and mz-aug2-50, it was not possible to compute the group of automorphisms within the chosen time limit (72 hours), even though they are small enough for the QAOA energy to be easily computable.}
\label{tab:dflarge}
\end{table*}

\begin{table*}
\centering{\scriptsize
\begin{tabular}{lrrrrrrrrrrrrr}
\toprule
\multirow{2}{*}{Name} &     \multirow{2}{*}{$|E|$} &    \multirow{2}{*}{$|V|$} & \multirow{2}{*}{$N_\textrm{orb}$} & \multirow{2}{*}{$t_\textrm{aut}$} &  \multicolumn{3}{c}{$p=1$}  &  \multicolumn{3}{c}{$p=2$}  & \multicolumn{3}{c}{$p=3$}  \\
\cmidrule(lr){6-8} \cmidrule(lr){9-11} \cmidrule(lr){12-14}
     &           &          &        &           & $t_\textrm{s}$        & $t_\textrm{acc}$  &  $S$       &         $t_\textrm{s}$ &            $t_\textrm{acc}$ & $S$ &  $t_\textrm{s}$  & $t_\textrm{acc}$        &  $S$    \\
\midrule
rnd-3-reg-10000-1 &   15000 &  10000 &     15000 &            113.65 &          154.55 &               156.91 &        0.57 &          421.05 &               423.59 &        0.78 &         1050.73 &              1036.18 &        0.91 \\
rnd-3-reg-3000-1  &    4500 &   3000 &      4500 &              3.38 &           46.15 &                45.89 &        0.94 &          126.29 &               124.46 &        0.99 &          313.86 &               310.65 &        1.00 \\
cfi-80            &    1200 &    800 &       360 &              0.28 &           12.32 &                 3.82 &        3.01 &           33.74 &                10.34 &        3.18 &           80.60 &                24.76 &        3.22 \\
cfi-20            &     300 &    200 &        90 &              0.04 &            3.16 &                 0.93 &        3.28 &            8.36 &                 2.56 &        3.22 &           20.30 &                 6.18 &        3.26 \\
grid-w-2-100      &   20000 &  10000 &         1 &             58.70 &          262.37 &                 0.02 &        4.47 &          827.07 &                 0.06 &       14.08 &         4960.41 &                 0.41 &       83.92 \\
mz-aug-22         &    1012 &    440 &       210 &              0.13 &           14.51 &                 3.09 &        4.51 &           41.53 &                 8.63 &        4.74 &          561.60 &                86.05 &        6.52 \\
mz-50             &    1500 &   1000 &       276 &              0.48 &           15.36 &                 2.89 &        4.56 &           42.16 &                 7.76 &        5.11 &          101.13 &                18.79 &        5.25 \\
mz-18             &     540 &    360 &       100 &              0.08 &            5.63 &                 1.12 &        4.68 &           14.86 &                 2.84 &        5.09 &           36.37 &                 6.85 &        5.25 \\
grid-3-20         &   22800 &   8000 &       550 &             36.96 &          406.37 &                10.04 &        8.65 &         1868.89 &                46.24 &       22.46 &             nan &                  nan &         nan \\
grid-w-3-20       &   24000 &   8000 &         1 &             47.57 &          441.59 &                 0.02 &        9.28 &         2124.36 &                 0.09 &       44.57 &             nan &                  nan &         nan \\
ag2-16            &    4352 &    528 &         1 &              0.37 &          215.36 &                 0.05 &      515.98 &         5448.04 &                 1.29 &     3290.06 &             nan &                  nan &         nan \\
pg2-32            &   34881 &   2114 &         1 &              3.45 &         3572.95 &                 0.10 &     1006.90 &        4.22e+05 &                12.16 &    27009.93 &             nan &                  nan &         nan \\
k-70              &    2415 &     70 &         1 &              0.14 &          399.55 &                 0.24 &     1036.67 &             nan &                  nan &         nan &             nan &                  nan &         nan \\
ag2-49            &  120050 &   4851 &         1 &             17.50 &        19600.24 &                 0.16 &     1109.76 &        8.05e+06 &                71.44 &    90539.60 &             nan &                  nan &         nan \\
k-100             &    4950 &    100 &         1 &              0.39 &         1306.11 &                 0.25 &     2024.96 &             nan &                  nan &         nan &             nan &                  nan &         nan \\
lattice-30        &   26100 &    900 &         1 &              1.63 &         4309.46 &                 0.17 &     2394.90 &             nan &                  nan &         nan &             nan &                  nan &         nan \\
latin-30          &   39150 &    900 &         7 &              1.09 &        11278.93 &                 2.11 &     3524.22 &             nan &                  nan &         nan &             nan &                  nan &         nan \\
paley-461         &   53015 &    461 &         1 &              0.41 &        64749.50 &                 1.32 &    37410.84 &             nan &                  nan &         nan &             nan &                  nan &         nan \\
\bottomrule
\end{tabular}}
\caption{Benchmark graphs with small search trees and thus easy for graph automorphism solvers. See the caption of Table~\ref{tab:dflarge} for definitions of column headers.}
\label{tab:dfsmall}
\end{table*}

\section{Discussion}\label{sec:discussion}

In this paper we presented an approach for accelerating the computation of the QAOA energy by leveraging fast graph symmetry solvers. We observed a median speedup of $1.81$ in the worst-case scenario of graphs specifically selected for being hard for automorphism solvers and $p=1$. Speedups are expected to be orders of magnitude larger in realistic settings where the QAOA energy is evaluated repeatedly, a situation that arises, for example, when seeking optimal QAOA parameters. 
When a slowdown---corresponding to situations where the automorphism calculation required more time than was saved by exploiting any found symmetry---was observed, the effect was generally far less extreme than the benefit encountered in other cases.  
In practice, we envisage that the automorphism solvers would be run with some small time limit (e.g., 10 minutes), with the QAOA energy computed the standard way if the solver did not successfully finish within this limit.

An important consideration for applicability of the proposed approach is the expected amount of symmetry in the problem, which is often known a priori. If the problem considered is defined on a random  three-regular or Erdos--Renyi model graph, then it is expected to have no symmetry~\cite{kim2002asymmetry,erdHos1963asymmetric}, and the step of computing symmetries can be omitted. At the same time, if the problem is defined on a structured graph, a significant amount of symmetry may be expected, and the proposed approach should be applied. As illustrated here, applying such an approach has little downside relative to the speedups it typically produces.

\rev{A sufficiently large symmetry group may make the problem classically trivial, which is the case for the complete graph $K_n$. On the other hand, in most cases taking advantage of symmetry to solve the problem is highly nontrivial, and symmetry may in fact make the problem harder to solve in practice. For example, in integer programming, symmetry is known to make problems extremely difficult for branch-and-bound solvers. There have been efforts to use information about the symmetry group to reduce the size of the search tree in branch-and-bound methods~\cite{ostrowski2011orbital,Margot2009}. However, these typically address only a subset of symmetries present in the problem and still lead to nontrivial running times for classical solvers even on moderately sized instances.}

An important future direction is focusing on local, rather than global, symmetries, as well as symmetries that are approximate, rather than exact. A step in that direction has been made in \cite{shaydulinsymm}. Fully exploiting the locality of QAOA and its interplay with local symmetries in the problem is expected to lead to significant further gains.

\section{Acknowledgments}

We thank Danylo Lykov for help with QTensor and Stuart Hadfield for the feedback on the manuscript. 
We are grateful to anonymous referees whose suggestions improved the presentation. 
This work was supported in part by the U.S.\ Department of Energy (DOE), Office of Science, Office of Advanced Scientific Computing Research AIDE-QC and FAR-QC projects and by the Argonne LDRD program under contract number DE-AC02-06CH11357. Clemson University is acknowledged for generous allotment of compute time on the Palmetto cluster.

\bibliographystyle{myIEEEtran}
\bibliography{references}

\newpage

\small

\framebox{\parbox{\textwidth}{
The submitted manuscript has been created by UChicago Argonne, LLC, Operator of 
Argonne National Laboratory (``Argonne''). Argonne, a U.S.\ Department of 
Energy Office of Science laboratory, is operated under Contract No.\ 
DE-AC02-06CH11357. 
The U.S.\ Government retains for itself, and others acting on its behalf, a 
paid-up nonexclusive, irrevocable worldwide license in said article to 
reproduce, prepare derivative works, distribute copies to the public, and 
perform publicly and display publicly, by or on behalf of the Government.  The 
Department of Energy will provide public access to these results of federally 
sponsored research in accordance with the DOE Public Access Plan. 
http://energy.gov/downloads/doe-public-access-plan.}}

\end{document}